\documentclass{llncs}
\usepackage{amsmath,amssymb,amsfonts}
\usepackage{graphicx}
\usepackage{algorithmic}
\usepackage{algorithm}
\usepackage{sober}

\begin{document}

\mainmatter  

\title{Fast phylogenetic tree reconstruction using locality-sensitive hashing }
\titlerunning{Fast phylogenetic tree reconstruction using locality-sensitive hashing}
\author{Daniel G. Brown \and Jakub Truszkowski}
\authorrunning{Daniel G. Brown and Jakub Truszkowski}
\institute{David R. Cheriton School of Computer Science\\
University of Waterloo\\
Waterloo ON N2L 3G1 Canada\\
\email{{browndg,jmtruszk}@uwaterloo.ca}}

\maketitle

\begin{abstract}
We present the first sub-quadratic time algorithm that with high probability correctly reconstructs phylogenetic trees for short sequences generated by a Markov model of evolution.  Due to rapid expansion in sequence databases, such very fast algorithms are becoming necessary.  Other fast heuristics have been developed for building trees from very large alignments~\cite{fasttree,WABI}, but they lack theoretical performance guarantees.
Our new algorithm runs in $O(n^{1+\gamma(g)}\log^2n)$ time, where $\gamma$ is an increasing function of an upper bound on the branch lengths in the phylogeny, the upper bound $g$ must be below$\frac{1}{2}-\sqrt{\frac{1}{8}} \approx 0.15$, and $\gamma(g)<1$ for all $g$. For phylogenies with very short branches, the running time of our algorithm is close to linear. For example, if all branch lengths correspond to a mutation probability of less than $0.02$, the running time of our algorithm is roughly $O(n^{1.2}\log^2n)$. Via a prototype and a sequence of large-scale experiments, we show that many large phylogenies can be reconstructed fast, without compromising reconstruction accuracy.
\end{abstract}
\newpage
\section{Introduction}

Phylogenetic reconstruction is a core bioinformatics problem.  Existing algorithms for the problem often require very long running times, particularly when the number of sequences is large; this problem is especially acute for traditionally slow, yet accurate, methods like maximum likelihood.  As biologists start to need trees for hundreds of thousands of taxa~\cite{greengenes}, even traditionally faster distance methods, such as neighbour joining, become unacceptably slow.  Researchers have developed sub-quadratic time heuristics~\cite{fasttree,WABI}, but they lack theoretical performance guarantees, so it is unclear whether their use is universally appropriate.

Here, we give an algorithm to correctly reconstruct, with high probability, phylogenetic trees that come from a Markov model of evolution, in  sub-quadratic time using sequences of $O(\log n)$ length. King has proved an $\Omega(n^2\frac{\log \log n}{\log n})$ lower bound on the running time of all distance-based phylogeny reconstruction algorithms using such short sequences; our algorithm avoids the problem by using the sequences directly, rather than relying only on distance calculations. 

Our algorithm is based on three ideas.  First, we use locality-sensitive hashing\cite{IndykMotwani} to find sequences that are near-neighbours in the tree, in sublinear time.  This hashing is a first step in choosing which two positions should be joined in the tree we incrementally are building.  Second, we use reliable estimates of distance, to identify exactly the correct join at each step; this step involves some hoary computation, due to the need to ensure that inferred sequences are independent estimates.  And finally, we use ancestral sequence reconstruction to reliably approximate the sequences found at internal tree nodes.  Since we start with a forest with each taxon in its own tree, and perform this joining step until only one tree remains in the forest, the overall runtime is sub-quadratic.  Specifically, if $p$ is an upper bound on the mutation probability on any edge, and $p<1/2-\sqrt{1/8}$, then we show that we can do the locality-sensitive hashing, which is the runtime-determining step, at each step in $O(n^{\gamma(p)} \log^2 n)$ time, where $\gamma(p)$ is always less than 1; the overall runtime is thus $O(n^{1+\gamma(p)}\log^2n)$.

\section{Related work}

Our work spans two recent threads in phylogenetic research: theoretical algorithms with guarantees of performance, and practical algorithms with no guarantees of performance.  

\subsection{Principled phylogenetic algorithms}
Erd{\H o}s {\em et al.}~\cite{erdos} gave an $O(n^4\log n)$ algorithm that reconstructs a phylogeny with high probability, assuming the Cavender-Farris model of evolution, for sufficiently long sequences. For most trees, their algorithm runs in $O(n^2\text{poly} \log n)$ time and requires $O(\text{poly} \log n)$ sequence length. Cs{\H u}ros~\cite{csuros} provided a $O(n^2)$ algorithm with similar performance guarantees. Recent papers~\cite{gronau,mossel} give similar algorithms to identify parts of the tree that can be reconstructed. These approaches use quartet queries chosen so that, with high probability, only correct quartets are queried. The only sub-quadratic time algorithm with guarantees on reconstruction accuracy is by King {\it et al.}~\cite{king}; for most trees, its running time is $O(n^2\frac{\log \log n}{\log n})$ provided that the sequences are $O(\text{poly} \log n)$ in length.

King {\it et al.} also showed that any algorithm that reconstructs the true tree with high probability, and uses distance calculations as its only source of information about the phylogeny, will have $\Omega(n^2\frac{\log \log n}{\log n})$ running time, for sequences of length $O(\text{poly}\log n)$.

Mossel~\cite{MosselPhase} gave a {\it phase transition} for phylogenetic reconstruction.  Suppose we have a balanced phylogeny with all edges having mutation probability $p$.  If $p$ is less than $\frac{1}{2}-\sqrt{\frac{1}{8}} \approx 0.15$, then the phylogeny can be reconstructed correctly from sequences of length $O(\log n)$ sequence data. For larger $p$, sequences must be of polynomial length to allow constant probability of reconstructing the phylogeny. Daskalakis, Mossel and Roch~\cite{Roch} extended the phase transition result to unbalanced trees, and provided an $O(n^5)$ algorithm that reconstructs phylogenies with lengths below the phase transition. Mihaescu, Hill and Rao~\cite{Mihaescu} provided an $O(n^3)$ algorithm for the same problem, which bears some resemblance to our approach.

\subsection{Practical algorithms}

In parallel to these theoretical results, many researchers have developed phylogeny reconstruction algorithms that can analyze alignments of tens, or even hundreds, of thousands of taxa. Fast Neighbour Joining~\cite{elias} runs in $O(n^2)$ time and gives results similar to neighbour joining, which requires $O(n^3)$ time. FastTree~\cite{fasttree} reconstructs phylogenies without computing the full distance matrix, which results in $O(n^{1.5}\log n)$ runtime. Our own recent algorithm, QTree~\cite{WABI,Almob}, runs in $O(n\log n)$ time, using an incremental approach to building trees. No theoretical guarantees exist for the quality of solutions obtained from these fast algorithms, however, under realistic assumptions about the evolutionary model, for short sequences.

\section{Preliminaries}

\subsection{Phylogenetic trees}
A phylogeny is an unrooted, weighted tree whose internal vertices all have degree $3$ and whose leaves represent extant taxa. The weights represent evolutionary time. Evolution is modelled as a time-reversible Markov process operating on the edges of the tree, where each position of a sequence evolves independently of all others. 

A {\it quartet} is a phylogeny on four taxa. For a set $\left\{ a,b,c,d \right\}$ of taxa, there are three possible quartets which we denote as $ab|cd$,$ac|bd$ and $ad|bc$.

We assume the Cavender-Farris model of binary character states over $\pm 1$ evolving according to a continuous Markov process.  Each edge $e$ is labelled with length $\ell(e)$, and the probability that the ends of $e$ have different states is $p(e)=\frac{1}{2}[1-\exp(-2\ell(e))]$. If two sequences  differ in a $\hat{p}$ fraction of sites, the maximum likelihood estimator of the distance between them is $\hat{d}=-\frac{1}{2}\log(1-2\hat{p})$. 

We assume there exist constants $f$ and $g$ such that for each edge $e$ in the phylogeny, we have $f<\ell(e)<g<\ln 2/4$.  This gives a minimum length for each edge, and also gives each edge state change probability less than $1/2-\sqrt{1/8}$, which guarantees a bounded probability of error when reconstructing ancestral sequences~\cite{Roch}.  We also assume that all edge lengths are multiples of some constant $\Delta$, consistent with previous work~\cite{Roch}. With these assumptions in place, a surprising fact arises: with sequences of length $O(\log n)$, we can exactly identify the tree distance between close nodes in the phylogeny~\cite{Roch,Mihaescu}.

\begin{theorem}
\label{dist-concentration}
Let $\Delta \leq f'<g'<\infty$. Then there exists a constant $c(f',g',\Delta)$ such that if $f'<d(a,b)<g'$ and $d(a,b)$ is a mutliple of $\Delta$, we have
$$
\Pr[|d(a,b)-\hat{d}(a,b)|>\frac{\Delta}{2}] \leq \exp(-k/c(f',g',\Delta))
$$where $k$ is the sequence length.

In particular, if $k = 3c(f',g',\Delta)\ln n$, we can identify the correct distance with probability at least $1-n^{-3}$.
\end{theorem} 

Note also that this theorem applies to any distances in our trees below a constant times $g$, the upper bound on a single edge length.


\subsection{Locality-sensitive hashing}

Our algorithm requires finding pairs of sequences within a specified small distance from each other, without having to compute all pairwise distances.  Indyk and Motwani~\cite{IndykMotwani} solved this problem using a collection of randomized hash tables: enough hash tables are chosen so that close sequences  likely collide in one of the tables, while keys are long enough that distant sequences do not.  This idea, known as {\it locality-sensitive hashing}, has been applied to many problems in bioinformatics, such as motif finding ~\cite{Buhler}.

Specifically, Indyk and Motwani solve a related problem, the $(r_1,r_2)$-approximate Point Location in Equal Balls ($(r_1,r_2)$-PLEB): 

{\bf Input}: A set of sequences $P$ in $\left\{ 0,1 \right\}^d$, a query sequence $q$, and radii $r_1<r_2$

{\bf Output}: If there exists a sequence $p \in P$ within normalized Hamming distance $r_1$ from $q$, output ``yes'' and a sequence within $r_2$ of $q$. If there is no sequence in $P$ within normalized Hamming distance $r_2$ from $q$, output ``no''. Otherwise, output either ``yes'' or ``no''.

Indyk and Motwani's solution constructs $n^{r_1/r_2}$ hash tables, each keyed on $O(\log n)$ randomly chosen sequence positions. Given $q$, a point within distance $r_1$ of it has a constant probability of colliding with it $q$ each hash table,  while points further than $r_2$ from $q$ have $O(1/n)$ probability of colliding. After inspecting a constant number of collisions with $q$, we can find, with constant probability, a point whose distance from $q$ is at most $r_2$; if we boost by running $O(\log n)$ times independently, the success probability is $1-n^{-\alpha}$, for any choice of $\alpha$. For more details, see~\cite{IndykMotwani}. Overall, finding an $(r_1,r_2)$-approximate near neighbour for a query point $q$ with high probability takes $O(n^{r_1/r_2} \log n)$ hash table lookups, each on a key of length $O(\log n)$ bits.

The reason Indyk and Motwani solve the approximate PLEB problem is that their hash table solution may be overwhelmed by points within the region greater than $r_1$ but less than $r_2$ away from $q$.  In our domain, we can avoid this problem, and find all sequences within distance exactly $r$ from a given query: we choose $r_2$ to be small enough that at most $O(\log n)$ points are found within even the $r_2$ distance, so we can examine all of them and still have fast runtimes.  We do this by choosing $r_2$ to be $1/2 - 1/2(\exp(-2cf\log \log n)$, the relative Hamming distance corresponding to all sequences within evolutionary distance $cf \log \log n$, for a constant $c$ that incorporates errors arising from reconstructing internal sequences of the tree (see Section 3.3 and the Appendix for details).  In $\log \log n$ edges, we can reach $O(\log n)$ nodes.  The distance $r_2$ converges to $1/2$ as $n$ grows (though it is quite a bit smaller for smaller values of $n$), so in the limit, the number of hash tables grows to $n^{2r_1}$.

Finding all neighbours within $r_1$ normalized Hamming distance thus takes $O(n^{2r_1+\epsilon}\log^2 n)$ time with high probability, where $\epsilon \rightarrow 0$ as $n$ increases: we use $O(n^{2r_1+\epsilon}\log n)$ hash tables, each of which requires $O(\log n)$ time to examine, and we take $O(\log ^2n)$ time examining the hash table hits.

In what follows, we will use the hashing algorithm to find sequences within evolutionary distance $d$, implicitly relying on the simple correspondence between Hamming and evolutionary distances outlined in the previous section: our procedure $FindAllClose(q,d)$ finds all sequences within evolutionary distance $d$ of $q$ with probability $1-o(1/n^3)$, and so with high probability makes no errors during the course of running our entire algorithm.

\subsection{Four-point method}

To identify the correct place to join two trees, we will use the four-point method, which reconstructs quartets from the six pairwise distance estimates. The method computes
$\hat{d}(a,b)+\hat{d}(c,d),\hat{d}(a,c)+\hat{d}(b,d)$, and $\hat{d}(a,d)+\hat{d}(b,d)$, and outputs $ab|cd$, $ac|bd$ or $ad|bc$, respectively, depending on which is the minimum.  If all pairwise distances were estimated exactly, the two sums corresponding to incorrect topologies would both be $2\ell(e)$ greater than the sum corresponding to the correct topology, where $e$ is the middle edge of the true quartet.

Because in our setting, we can estimate distances exactly with high probability, we can assume that all quartets are properly computed, provided that the sequences used to compute them satisfy the {\it error-independence} property, defined in the next subsection.

\begin{theorem}
\label{thm-four-point}
Let $f$ and $g$ be the upper and lower bounds on the edge length in a quartet tree. Then there exists a constant $c_2(f,g,\Delta)$ such that we can reconstruct the lengths of the edges of the quartet exactly using sequences of length $c_2(f,g,\Delta)\log n$ with probability at least $1-\frac{6}{n^3}$.
\end{theorem} 

\begin{proof}
The claim follows from Corollary 1, by setting $f'=2f$ and $g'=3g$.
\end{proof}

\subsection{Ancestral states}

When all edge lengths are below the phase transition threshold, we can correctly infer the ancestral state of a character at any internal node of the tree with probability greater than $\frac{1}{2}+\beta$ for some constant $\beta$~\cite{Kenyon,MosselPhase,Roch}. Using this observation, we identify nodes which should be near-neighbours in the tree, join them together, and infer new ancestral node sequences, until we have only a single tree remaining.  The following theorem, which is an adaptation of a result by Daskalakis {\it et al.}~\cite{Roch}, describes the probability of correctly reconstructing ancestral states.

\begin{theorem}
\label{anc-reconstruction}
Let $T$ be a binary tree with root $\rho$ and edges $e$ all satisfying $p(e)<\frac{1}{2}-\sqrt{\frac{1}{8}}$.  Let $\delta T$ denote the leaves of $T$. Let $\sigma$ be a Cavender-Farris character on $T$.  The maximum likelihood algorithm $A$ for ancestral state reconstruction computes the ancestral state with probability satisfying
$$
\Pr[A(\sigma_{\delta T})=\sigma_\rho|\sigma_\rho=1]=\Pr[A(\sigma_{\delta T})=\sigma_\rho|\sigma_\rho=-1]>\frac{1}{2}+\beta.
$$for some constant $c$ independent of $T$.
\end{theorem}

The maximum likelihood algorithm~\cite{Felsenstein} computes the posterior distribution of a state at an internal node based on the previously computed posterior distributions at its children. For constant-sized alphabets, this can be done in $O(1)$ time. We then pick the state of largest posterior probability. If the true edge lengths are known, this algorithm has the optimal probability of correctly reconstructing the ancestral state, among all possible algorithms.

The proof of~\cite{Roch} used the {\it recursive majority algorithm}, making the parameter $\beta$ hard to estimate. We need an upper bound on $\beta$ to bound the running time of our algorithm. The following result was proved by Evans {\it et al.}~\cite{Kenyon}. 

\begin{theorem}
\label{evans}
Let $T$ be a phylogenetic tree where all mutation probabilities across edges are bounded by $p_q<\frac{1}{2}-\sqrt{\frac{1}{8}}$. Assign to each edge $e$ a resistance $(1-2p)^{-2|e|}$, where $|e|$ is the number of edges on the path from root to $e$, including $e$ itself. The probability $p_{err}$ of incorrectly reconstructing the state at the root of the tree is bounded by

$$
p_{err}<\frac{1}{2}-\frac{1}{1+\mathcal{R}_{\mathit{eff}}}
$$where $\mathcal{R}_{\mathit{eff}}$ is the effective resistance between the root and the leaves of $T$. 
\end{theorem}

This bound is quite loose. For this reason, we will use a better bound, originally developed for the simpler Fitch parsimony algorithm. The following bound is sharper for $p_q<0.118$.

\begin{theorem}
\label{fitch-rec}
Let $T$ be a phylogenetic tree where all mutation probabilities across edges are equal to $p_q<\frac{1}{8}$. The probability $p_{err}$ of incorrectly reconstructing the state at the root of the tree using  Fitch parsimony is bounded by

$$
p_{err}<\frac{1}{2}-\frac{\sqrt{(1-4p_g)(1-8p_g)}}{2(1-2p_g)^2}<1-4p_g
$$
\end{theorem}

Note that the original result was stated to hold for variable edge lengths. This was corrected by Zhang {\it et al.}~\cite{Zhang}.

The above bound applies to the maximum likelihood algorithm, even for variable edge lengths, as long as they are bounded by $p_q$. For constant length edges, the result holds by the optimality of maximum likelihood. Now suppose that we shrink an edge in tree $T$, creating a tree $T'$. The mutual information between the leaves and the root of $T'$ is greater than the mutual information between the leaves and the root of $T$, meaning that the probability of reconstructing the root of $T'$ must be higher. Applying this argument to all edges except the longest, we obtain that the above bound holds for trees of variable edge lengths when the maximum likelihood algorithm is used.

Knowing the bound on $p_{err}$ is important in our algorithm because we will use locality-sensitive hashing on these sequences, and we need to know the number of hash tables required. Better bounds on $p_{err}$ will result in faster algorithms with the same performance characteristics.  Let $g_{err}$ be the distance corresponding to a mutation probability of $p_{err}$.

\begin{figure}[t]
\centering
\includegraphics[width=5in]{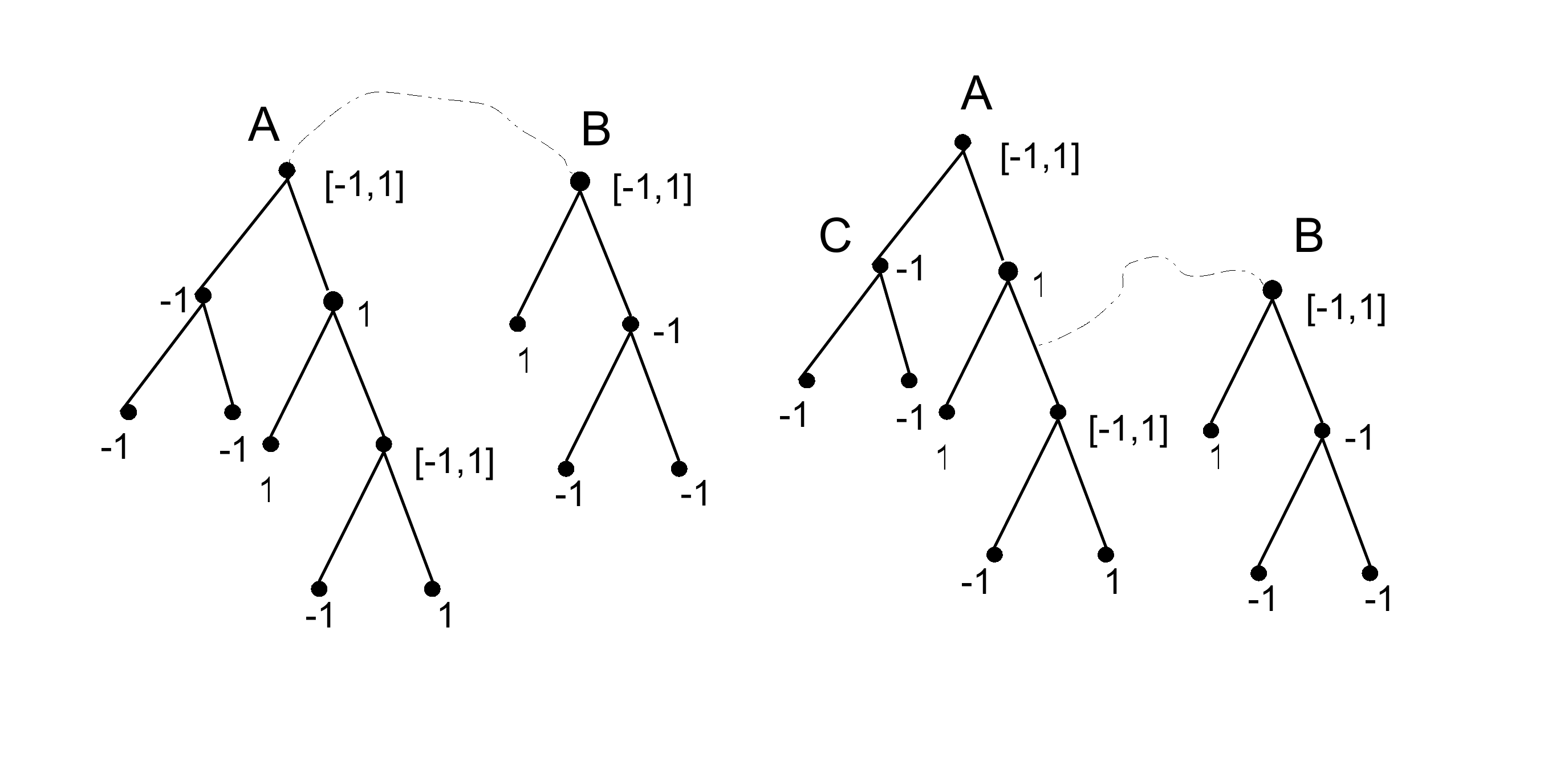}
\caption{Independence relationships between sequences reconstructed using Fitch's algorithm: in the left tree, sequences $A$ and $B$ are error-independent. In the right tree, $A$ and $B$ are not error-independent, but $C$ and $B$ are error-independent. The algorithm was run on subtrees drawn in solid lines.
\label{independent-reconstruction}}
\end{figure}

Suppose that we reconstruct ancestral sequences for subtrees $T_1,T_2$ of $T$, rooted at $\rho_1,\rho_2$, respectively. Moreover, suppose that the path connecting $T_1$ and $T_2$ has ends $\rho_1,\rho_2$ (see Figure~\ref{independent-reconstruction}). By the Markov property, reconstructing $\sigma_{\rho_1}$ correctly is independent of reconstructing $\sigma_{\rho_2}$ correctly. We call such two sequences {\it error-independent}.  The distance estimate $\hat{d}(\sigma_{\rho_1},\sigma_{\rho_2})$ will converge to $g+g_1+g_2$, where $g_1$ and $g_2$ are edge lengths corresponding to the probabilities of incorrectly reconstructing states in the two sequences. When comparing independently reconstructed sequences, we can effectively treat these errors in the reconstructed sequences as ``extra edges''~\cite{Roch}, whose length can be bounded using Theorem 4. This observation has been used extensively in many theoretical algorithms~\cite{Roch,MosselPhase,Mihaescu}. 

Theorem~\ref{anc-reconstruction} combined with Theorem~\ref{thm-four-point} and the Markov property enable us to estimate internal branch lengths from reconstructed ancestral sequences, and also correct quartets.

\begin{theorem}
\label{anc-quartet}
Let $i,j,k,l$ be ancestral sequences reconstructed by maximum likelihood from four disjoint subtrees, and such that no path between any two of them in the true phylogeny shares an edge with any of the subtrees. Let $f'$ and $g'$ be the upper and lower bounds on the edge lengths in the quartet $ijkl$. Then there exists a constant $c(f',g',\Delta)$ such that given reconstructed ancestral sequences of length $c\log n$, we can estimate the length of the middle edge of $ij|kl$ exactly, with probability  $1-\frac{6}{n^3}$. 
\end{theorem} 

\section{The algorithm}

The algorithm starts with a forest $F$ of $n$ trees, each with one taxon. It then progressively merges subtrees into larger subtrees of the true tree, by finding two nodes that are quite close using locality-sensitive hashing, identifying where they should be joined, and inferring ancestral sequences.  The underlying idea, however, is complicated by the requirement in Theorem~\ref{anc-quartet} that the sequences be reconstructed from disjoint subtrees of the true phylogeny.

Algorithm \ref{simple} presents a simplification of the algorithm.  Our algorithm maintains four invariants:

\begin{enumerate}
\item Every tree in $F$ is a subtree of the true tree $T$
\item No two trees in $F$ overlap as subtrees of $T$
\item For each tree $T'$ in $F$, all edges except at most one have length at most $g$. The remaining edge has length at most $2g$. We call it a {\it long edge}.
\item The length of every path in every subtree in $F$ is reconstructed correctly
\end{enumerate}

The first three invariants are the same as in the work of Mihaescu {\it et al.}~\cite{Mihaescu}; the fourth, we maintain using Theorem \ref{anc-quartet}.
The invariants, together with routine $CheckErrorIndependence$, ensure that its preconditions are satisfied. 

\begin{algorithm}
\caption{SimplifiedReconstruct($\{\sigma\},f,g$)}
\label{simple}
\begin{algorithmic}
\STATE Start with a forest with each node in its own tree.
\STATE Use locality-sensitive hashing to find all sequences whose pairwise distance is less than $3g+2g_{err}$. Put them in a priority queue, \emph{DistQueue}.
\WHILE{the forest has more than one tree}
\STATE Find two sufficiently close nodes $x$ and $y$ that are not currently in the same tree of $F$.
\STATE Identify the nearby edge $(i,j)$ to $x$ and $(k,l)$ to $y$ that should be joined to connect the trees containing $x$ and $y$ in the forest
\STATE Create two new nodes, $a$ and $b$ in the middle of $(i,j)$ and $(k,l)$, and join $a$ and $b$ together with a new edge.  
\STATE Estimate the lengths of all five edges in the quartet $ijkl$.
\STATE Reconstruct the ancestral sequences at $a$ and at $b$.
\STATE Find all sequences within $3g+2g_{err}$ of the newly inferred sequences, and add these distances to \emph{DistQueue}.
\ENDWHILE
\end{algorithmic}
\end{algorithm}

In what follows, we will expand the details of this algorithm, focusing on ensuring invariants 3 and 4.  We will assume the existence of three procedures: $FindAllClose(q,d)$ uses locality-sensitive hashing to identify all sequences within relative Hamming distance $d$ of $q$.  $Quartet(a,b,c,d)$ uses the four-point method to identify the correct topology of the quartet $abcd$.  And $MiddleEdge(ab|cd)$ computes the length of the middle edge in the quartet $ab|cd$.  Assuming the preconditions to Theorem \ref{anc-reconstruction}, these procedures work with high probability.

Most of the subroutines are presented for the case where all their arguments are internal nodes. The cases where some nodes are leaves are analogous; we omit them for brevity.

We will often treat subtrees with long edges as rooted, with the root located somewhere on the long edge. 

\subsection{Independent inferences}

If the reconstructed sequences used to perform quartet queries are not independent, the quartet middle edge length estimates and the inferred quartet topology might be incorrect. This could lead to a wrong choice of which edges to join.


The order in which ancestral sequences are reconstructed defines a partial order of the nodes in $F$. We call it the {\it reconstruction order}. For two nodes $a$ and $x$ in $F$, at least one of the children of $b$ with respect to reconstruction order is not on the path from $a$ to $x$ in $T$. $CheckErrorIndependence$ uses this observation to detect cases when lack of error-independence impacts the middle edge length estimate for the quartet. 

\begin{algorithm}
\caption{CheckErrorIndependence$(x,y,a,b)$}
\label{a1}
\begin{algorithmic}
\REQUIRE $a$ and $b$ are error-independent, $x$ and $y$ are error-independent, $T(a)$ and $T(x)$ do not overlap
\FORALL{$z \in \left\{a,b,x,y \right\}$ }
	\STATE Let $z_1,z_2$ be the children of $z$ in reconstruction order(if they exist)
\ENDFOR
\STATE $d \gets MiddleEdge(xy|ab)$
\FOR{$i,j \in \left\{1,2\right\}$}
	\STATE $d_{a_ib_j} \gets MiddleEdge(xy|a_ib_j)$
	\STATE $d_{x_iy_j} \gets MiddleEdge(x_iy_j|ab)$
\ENDFOR
\IF{any of the $d_{a_ib_j}$ or $d_{x_iy_j}$ differs from $d$ by more than $\Delta/2$}
	\RETURN false
\ENDIF
\RETURN true
\end{algorithmic}
\end{algorithm}

Let $ME(xy|ab)$ denote the true length of the middle edge in the quartet $xy|ab$ in $T$.

\begin{lemma}
If $|MiddleEdge(xy|ab)-ME(xy|ab)|>\Delta/2$, then CheckErrorIndependence returns false. Otherwise, it returns true with high probability.
\end{lemma}
\begin{proof}
Let $T(a)$ and $T(x)$ be the subtrees that contain $a$ and $x$, respectively. Sequences $a$ and $b$ are independent and lie on the opposite sides of edge $e$. It follows that if $x$ and $y$ join the tree at $e$, we have $ME(xy|ab)=ME(xy|a_ib_j)$ for any choice of $i$ and $j$. If $x,y,a,b$ are independent, then all the middle edge estimates are within $\Delta/2$ of each other and correct within $\Delta/2$. Now suppose some two of these sequences are not independent (say $a$ and $x$). Without loss of generality, that means that $T(x)$ joins $T(a)$ at some edge in subtree of $T(a)$ consisting of all nodes based on which the sequence at $a$ was reconstructed. It is easy to see that one of the sequences $a_1,a_2$ is then independent of $x$, in which case its corresponding call to $MiddleEdge$ will return a value that is correct within $\Delta/2$.
\end{proof}

\subsection{Detecting overlapping subtrees}

To maintain Invariant 2, we need to prevent merging two trees if the new edge created between them overlaps some other subtree in $F$. The procedure $CheckOverlaps$ detects overlapping subtrees. We use $R(T(x))$ to indicate the set of sequences in $T(x)$ and $d_{T'}$ for the path metric associated with $T'$.

\begin{algorithm}
\caption{CheckOverlaps(x,y)}
\label{a1}
\begin{algorithmic}
\STATE $S=(FindAllClose(x,2g+2g_{err}) \cup FindAllClose(y,2g+2g_{err})) - R(T(x))$ 
\FOR{each sequence $a$ in $S$}
	\FOR{each sequence $b$ in $T(a)$ that is independent of $a$ and such that $d_{T(a)}(a,b)<5g$}
		\IF{$Quartet(a,b,x,y) \neq ab|xy$ and $CheckIndependence(Quartet(a,b,x,y))=true$}
					\RETURN true
		\ENDIF
	\ENDFOR
\ENDFOR
\RETURN false
\end{algorithmic}
\end{algorithm}

\begin{lemma}
If the edge $(x,y)$ overlaps some other edge in $F$ and the sequences at $x$ and $y$ are independent, $CheckOverlaps$ will return true. Otherwise, it will return false.
\end{lemma}
\begin{proof}
Suppose $T(x)$ overlaps with some other subtree $T'$. One can easily show by case analysis that there exists a reconstructed sequence  in $T'$ that is within distance $2g+2g_{err}$ from $x$ or $y$ and such that the two sequences are independent. Since $T(x)$ overlaps with $T'$, there has to exist a node $b'$ in $T'$ such that the induced topology on $x,y,a,b'$ is $ax|b'y$ (w.l.o.g. we assume that $a$ forms a clade with $x$. The reconstructed sequence at $b'$ need not be independent from $a$, $x$ and $y$, but then there must be a node $b$ within distance $2g$ from $b'$ whose reconstructed sequence is independent of the other sequences, which means $ax|by$ will pass the independence test. 
On the other hand, if $T(x)$ does not overlap with any other subtree, all quartet queries that pass the independence test will return $ab|xy$.
\end{proof}

Note that searching for $b$ can be done by breadth-first search on $T(a)$, without resorting to the nearest neighbour search algorithm.

\subsection{Three-way ancestral sequence reconstruction}

In order to connect edges from different subtrees, we need to have independent sequence reconstructions at both ends of the new edge. To ensure this can be achieved, we will maintain, where possible, three separate sequence reconstructions at each internal node of the subtree, each based on two subtrees of $T(x)$ created by removing $x$, but independent of the third subtree. For any edge $e$ in a subtree, we refer to the two independent sequences at its ends as {\it companions}.

Theorem \ref{anc-reconstruction} is only applicable when all subtree edges have length less than $g$.  For a subtree with a long edge, we treat the subtree as rooted at a node on the long edge with a single sequence reconstruction.  When that node joins to another tree, we then create new three-way reconstructions; such a tree can only be joined with another tree via the long edge, in order to maintain Invariant 3.

The routine $ThreeWayReconstruction$ takes a tree with sequences reconstructed by maximum likelihood, adding to each vertex the two remaining reconstructions of its sequence. It must be started from a node that has no successors in  reconstruction order.

\begin{algorithm}
\caption{ThreeWayReconstruction($r$)}
\label{3way}
\begin{algorithmic}
\STATE Let $x,y,z$ be the neighbours of $r$.
\STATE Reconstruct sequences $\sigma_{xy}(r),\sigma_{yz}(r),\sigma_{xz}(r)$ conditioned on $\left\{x,y\right\}$,$\left\{y,z\right\}$ and $\left\{x,z\right\}$, respectively.
\STATE Let $S$ be the set of vertices in $T(r)$ with only one sequence reconstruction.
\STATE Visit vertices in $T(r)$ in Breadth-First Search order, reconstructing each sequence conditioned on all choices of $2$ neighbours. Stop branching if a node visited during an earlier call of $ThreeWayReconstruction$ is encountered.
\end{algorithmic}
\end{algorithm}

\subsection{Long edges must be joined}
Again, to maintain Invariant 3, when either of the two closest sequences is found in a subtree with a long edges, we must break that edge.  Because we will be finding the shortest pairwise distance between trees in $F$, we will certainly find one of the endpoints of the long edge, but we must not consider its other neighbouring edges.

Procedure $CandidateEdges(x)$ identifies valid edges that can be broken, given a node $x$ in the tree.

\begin{algorithm}
\caption{CandidateEdges($x$)}
\label{candedges}
\begin{algorithmic}
\IF{$x$ is the root of a tree $T$ in $F$} 
\STATE Return the set containing the edge $e$ that contains $x$.  (This may be a long edge.)
\ELSE
\STATE Return the set containing all edges in $T$ within distance $3g+2g_{err}$.  (This can be determined by a breadth-first search.)
\ENDIF
\end{algorithmic}
\end{algorithm}


\subsection{The full algorithm}
With these minor issues resolved, we can present the full algorithm.  

\begin{algorithm}
\caption{Reconstruct($\left\{ \sigma \right\}$,$f$,$g$)}
\label{a1}
\begin{algorithmic}
\STATE Start with a forest $F$ where each sequence is in its own tree.
\STATE Initialize hash tables for nearest neighbour search
\STATE Use \emph{FindAllClose} to identify all sequence pairs at distance less than $3g+2g_{err}$; put these in a  queue DistQueue.
\STATE $WaitList \gets \emptyset$
\WHILE{$F$ has more than one tree}
\STATE $(x,y) \gets DistQueue.pop()$
\IF{$T(x)=T(y)$}
	\STATE continue
\ENDIF
	\STATE $X\gets CandidateEdges(x)$
	\STATE $Y\gets CandidateEdges(y)$
        \STATE $Joins \gets X \times Y$.
		\IF{$X=\emptyset$ or $Y= \emptyset$}
			\STATE $WaitList.add((x,y))$
		\ENDIF
        \STATE Filter $Joins$ to only include pairs $((i,j),(k,l))$ where $Independent(i,j|k,l)$ 
        \STATE Use \emph{Quartet} and \emph{MiddleEdge} to find $d_{min}$, the smallest middle edge length among $Joins$.  Let it be the result of joining $(i^*, j*^)$ with $(k^*, l^*)$.  
		\STATE create an edge between the two edges $(i^*,j^*)$ and  $(k^*,l^*)$ that give rise to $d_{min}$
		\STATE use $MiddleEdge$ to calculate the lengths $d_1,\ldots,d_5$ of the five new edges
		\IF{$\max_i d_i \geq 2g$ or $d_i>d_j>g$ for some $i \neq j$ or $CheckOverlaps(x,y)$}
			\STATE undo this loop iteration
			\STATE $WaitList.add((x,y))$
		\ENDIF
        \IF {the new tree has a long edge}
		\STATE create a root on the new long edge
                \ELSE
		\STATE reconstruct the sequence at  the new internal nodes $r_1,r_2$
		\ENDIF
		\IF{the new tree has no long edge}
			\STATE $ThreeWayReconstruction(r_1)$
		\ENDIF

        \STATE Use \emph{FindAllClose} to add all newly-created sequence pairs whose distances are below $3g+2g_{err}$ to \emph{DistQueue}.
		\STATE For all hits $(x,y)$ in \emph{WaitList} at distances less than $3g+2g_{err}$ from any of the newly created sequences, move $(x,y)$ from \emph{WaitList} to \emph{DistQueue}

\ENDWHILE
\end{algorithmic}
\end{algorithm}

\begin{lemma}
At any point during the execution of the algorithm, there always exists a pair of subtrees that can be merged.
\end{lemma}

\begin{proof}
There are two cases when two subtrees with internal nodes within distance $2g$ from each other cannot be merged. One is when one of the hits is not a root and $ThreeWayReconstruction$ has not been called on its subtree. The other is when merging two trees would give rise to a tree with two long edges. 

Let us consider the second case first. For each two subtrees $T_1,T_2$ with that property, remove them from $T$ together with the path connecting them. This gives rise to a forest $F_1$ where at least two components border exactly one of the removed subtrees. Call one such component $T'$ and let $T''$ be its unique adjacent removed subtree. Let $e'$ be the edge between $T'$ and $T''$.

If $T'$ contains a reconstructed subtree that is incident to $e'$ and can be joined with the subtree at the other end of $e'$, we are done. Otherwise, we will show that there exist two subtrees in $T'$ that can be merged. Consider the forest $F^c=T'-F$. We will refer to components of $F^c$ as {\it antitrees} to avoid confusion with reconstructed trees from $F$. It is easy to see that since $F$ cannot contain leaves that are internal nodes in $T$, each antitree in $F^c$ is either a single edge or contains a cherry. Therefore, each antitree in $F^c$ contains two leaves at distance less than $2g$. 

Take any such two leaves $x$ and $y$ from antitree $T^c_1$. If they cannot be joined, then at least one of the trees (say $T(x)$) has a long edge that doesn't include $x$. The root $r$ at this long edge belongs to an antitree $T^c_2$ that must be different from $T^c_1$ (otherwise we would have a cycle in $T$). $T^c_2$ has at least one pair of leaves within distance $2g$. If they cannot be merged, then one of its leaves is incident to another tree in $F$ with a long edge. Its root is incident to another tree $T_3^c$ in $F^c$. The claim holds by induction on the size of $F^c$.
\end{proof}


\begin{theorem}
Each iteration of the {\bf while} loop maintains invariants 1-4.
\end{theorem}

\begin{proof}[Sketch]
Invariants 1 and 4 are maintained due to Theorem~\ref{anc-quartet} and the fact that every quartet query is checked for independence. Invariant 2 is maintained by $CheckOverlaps$. Invariant $3$ is maintained by the conditions on $d_{min}$ in the main loop.
\end{proof}

\subsection{Runtime analysis}
\label{rtime}
At each iteration of the {\bf while} loop, all the operations except nearest neighbour search (which is invoked a constant number of times per iteration) take constant time, given that the ratio $g/f$ is constant(and thus the number of nodes within distance $g$ of any node is a constant whose size is $O(2^{g/f})$. The complexity is therefore dominated by the use of the hash tables. The evolutionary distance $3g+2g_{err}$ corresponds to a Hamming distance of at most
$$
h=\frac{1}{2}-\frac{1}{2}(1-2p_g)^3(1-2p_{err})^2
$$where $p_g=\frac{1-e^{-2g}}{2}$ and the bound on $p_{err}$ is given by Theorem~\ref{fitch-rec}. $FindAllClose$ is used a constant number of times per loop, so the running time of each iteration of the loop is $O(n^{2h+\epsilon}\log^2 n)$ . The loop is run $O(n)$ times, since the number of sequences within distance $3g+2g_{err}$ of any sequence is constant. Overall, the runtime is bounded by

$$
Cn^{2-(1-2p_g)^3(1-2p_{err})^2+\epsilon}\log^2 n < Cn^{2-(1-2p_g)^3\frac{(1-4g)(1-8g)}{(1-2g)^4}+\epsilon}\log^2 n<Cn^{2-(1-2p_g)^3(8g-1)^2+\epsilon}\log^2 n
$$which is always $o(n^2)$. Table 1 shows the runtime for selected values of $p_g$.

\begin{center}
\begin{table}
\caption{The approximate runtime for different values of $p_g$}
\centering 
\begin{tabular}{|c|c|}

		$p_g$ & runtime  \\ \hline
		0.01 & $n^{1.10}\log^2 n$  \\ \hline
      0.02 & $n^{1.19}\log^2 n$  \\ \hline
      0.05 & $n^{1.47}\log^2 n$   \\ \hline
		0.075 & $n^{1.67}\log^2 n$   \\ \hline
		0.10 & $n^{1.85}\log^2 n$   \\ \hline
		\end{tabular}
    \end{table}
\end{center}

\section{Experiments}
\subsection{A practical algorithm}

Many assumptions made by the above theoretical algorithm cannot be met in practice. The number of hash tables required for \emph{FindAllClose} to work with high probability requires a prohibitive amount of memory  on a standard desktop computer. Using maximum likelihood for ancestral sequence reconstruction also requires a large amount of memory to store conditional probabilities. This motivated us to develop a simpler and more memory-efficient practical algorithm.

Our implementation uses a number of hash tables required to find near neighbours with constant probability, not high probability. For reasons of memory efficiency, we also do not perform three-way reconstruction; instead, we join non-root nodes of different subtrees without requiring the sequences in quartet queries to be error-independent. Note that we still require error-independent sequences for estimating branch lengths in an existing subtree. After two subtrees are joined, the sequences in the smaller subtree are re-estimated according to an ordering compatible with that of the larger tree. 

For simplicity, the practical algorithm does not use routines \emph{CheckOverlaps} and \emph{CheckErrorIndependence}. Instead, we perform a small number of Nearest Neighbour Interchange (NNI) operations after each join, to ameliorate the problems originally addressed by these two functions. After an edge has been added, we re-estimate the length of all edges whose length might have been affected by the merge. If any quartet gives a topology that is not consistent with the edge, we perform an NNI operation to fix the topology and re-estimate the lengths of adjacent edges. While this process might repeat several times, it is not equivalent to traditional local search algorithms using NNI's, as only the edges in the close vicinity of the new edge are affected and the associated computational cost is much lower. 

The algorithm tries to merge pairs of trees, starting from collisions with the lowest estimated evolutionary distance. If no collisions are found within distance $(r_1+r_2)/2$, a new hash table is added, until the maximal number of $2n^{r_1/r_2}$ hash tables is reached.

The current version of our algorithm does not attempt to find optimal LSH parameters $r_1$ and $r_2$. In our experiments, we set them to $0.2$ and $0.6$, respectively. This choice leads to memory inefficiency for trees with very short edges, as we will see later. We leave the automatic adjustment of these values for future work.

\subsection{Data sets}

We used a data set from our previous paper, where we presented QTree~\cite{Almob}, to compare our new algorithm with two other fast phylogenetic algorithms, Neighbour Joining and FastTree. We simulated 10 trees on 20000 taxa from the pure-birth process. We then multiplied the length of each branch by a factor chosen uniformly at random from interval $[0.5,2]$ to deviate the trees from ultrametricity. This methodology follows the previous work of Liu {\it et al.}~\cite{SATE}. We then scaled the branch lengths of the entire tree by several choices of constant factors. For each choice of tree and scaling factor, we generated alignments whose length varied between 250 and 4000 positions. No indels were introduced in the simulation.


\subsection{Results}

Figure 2 shows the performance of our algorithm, compared with QTree and FastTree. We use the Robinson-Foulds metric, defined as the fraction of splits in the true tree that are found in the reconstructed tree. To ensure fair comparison, we only ran the Neighbour Joining phase of FastTree, without its concluding local search phase.

Our algorithm achieves higher accuracy than both FastTree and QTree in most settings. The main exception is trees with long branches, where its accuracy is substantially lower than both QTree and FastTree. The poor performance of our algorithm for trees with long branches is not surprising given that it relies so heavily on reconstructed ancestral sequences, whose accuracy diminishes as branch lengths approach the phase transition. For very short sequences, FastTree appears somewhat more accurate, possibly because of aggregating information from a greater number of distance estimates. 

\subsection{Running times and scalability}

On most instances, our program runs in times competitive with FastTree and somewhat longer than QTree (see Table 2). We believe the running times could be improved by a more careful choice of parameters, and note that our work is a preliminary prototype.

For 32-bit machines with up to 4GB RAM, our program does not scale to alignments larger than $2 \cdot 10^7$ letters. This is mostly due to the amount of memory required to store probability vectors for maximum likelihood, but also due to the hash tables. Memory usage may also increase if the number of collisions is high, since these are stored in a priority queue.

For trees where average branch length is very low compared to the $r_2$ parameter, vast numbers of collisions are generated, which leads to a substantial increase in running time and memory usage. We partially mitigate this problem by discarding all but top $k$ hits from each hash table entry, but the increase in running time is still substantial, sometimes increasing by $3$-fold compared to the normal scenario. We plan to solve this problem by supporting longer hash table keys and automatically choosing $r_2$ in the final version of the software.
\newpage
\begin{table}
\label{runtimescomparison}
\caption{The running times of the three algorithms for three representative data sets. In most cases, our algorithm is faster than FastTree, but slower than QTree. For very short branches, the number of hash table collisions is very high due to $r_2$ being too large, which results in a longer running time for our algorithm.}
\begin{tabular}{|c|c|c|c|}
		algorithm & $scale=25,seqlen=1000$& $scale=50,seqlen=1000$ & $scale=100,seqlen=1000$ \\ \hline
		QTree &  4m57s & 5m39s&  6m28s\\ \hline
      Our algorithm & 24m49s & 8m27s &  6m50s\\ \hline
      FastTree (NJ phase only) & 10m31s & 11m01s&  11m23s\\ \hline
		\end{tabular}
\end{table}
\begin{figure}[h!]
\vspace{-4cm}
\label{accgraphs}
\includegraphics[width=6in]{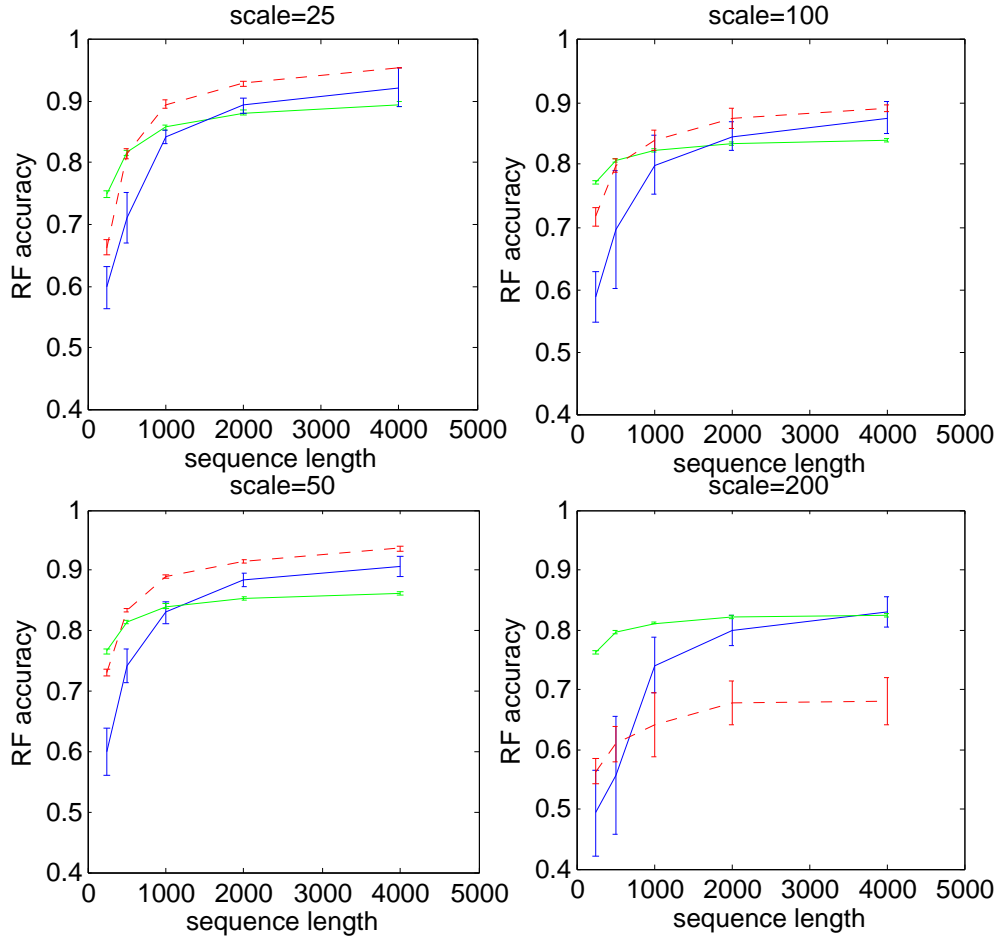}
\vspace{-3cm}
\caption{The performance of the LSH algorithm(red, dashed) compared to QTree(dark blue), and FastTree(light green),  as a function of the length of the sequences. The four graphs represent the performance on 10 tree topologies with branch lengths scaled by constant factors $25,50,100$, and $200$. The accuracy of the LSH algorithm is superior to both QTree and FastTree in most settings, except for phylogenies with very long branches (scale=200), where the LSH algorithm performs substantially worse than the other two, presumably due to poor ancestral sequence reconstruction.}
\end{figure}


We also ran our program on the larger simulated 16S data set with 78000 sequences from the FastTree paper~\cite{fasttree}. We created smaller data sets by randomly sampling 20000 and 40000 sequences from the full data set. For the data set with 20000 sequences, our algorithm took 15 minutes, compared with 9 minutes for both FastTree and QTree. For the data set with 40000 sequences, our algorithm took 56 minutes, compared with 26 minutes for FastTree and 19 minutes for QTree. We think that the runtime of our algorithm was impacted by the wrong choice of $r_2$, as in previous experiments. Our program ran out of memory on the full data set. The accuracies of the algorithms behaved similarly to the other data sets, with our algorithm being more accurate than QTree and FastTree.

\section{Conclusions and future work}
\label{conclusion}
We have presented a fast theoretical algorithm that correctly reconstructs phylogenies whose branch lengths are short enough. This theoretical algorithm shows the possibility of reconstructing such large phylogenies in sub-quadratic time from short sequences, without compromising the accuracy. Our prototype implementation achieves accuracies that are comparable or exceed existing algorithms, while also offering competitive running times for  instances of a few tens of thousands of taxa. We believe that both the accuracy and the running time of the algorithm could be improved further.

This work could be improved in several ways. On the practical side, we plan to improve the scalability of the algorithm by using a more flexible and memory-efficient hash table implementation. The applicability of our algorithm to diverse evolutionary scenarios will require the ability to set the hash table parameters automatically. We also plan to investigate how the runtime of LSH algorithms could be improved by taking advantage of rate variability across sites, which may also offer opportunities for higher accuracy in cases where long branches are present.

Some theoretical questions also remain. Felsenstein's algorithm is known to have optimal probability of reconstructing ancestral states correctly, but this probability appears hard to estimate (see e.g.~\cite{MaLouxin}). Getting tighter bounds on reconstruction accuracy  would lead to an improved running time of our algorithm. Another avenue for improvement is using a faster locality-sensitive hashing scheme. Dubiner~\cite{Dubiner} has recently proposed such a scheme for very long sequences, but it is not clear whether it will be useful with sequences of only logarithmic length. 

\bibliographystyle{splncs03}
\bibliography{full_quartets4}
\appendix
\section{The effect of reconstruction errors on locality-sensitive hashing}

If two reconstructed ancestral sequences are not error-independent, the distance estimate between them might be biased. Here, we show that this has no effect on the running time of \emph{FindAllClose} and its ability to find all sequences within specified evolutionary distance. The following lemma is a direct consequence of Lemma 5.4 by Mihaescu {\it et al.}~\cite{Mihaescu}.

\begin{lemma}
If nodes $a$ and $b$ are not error-independent and their true evolutionary distance is $d$, $\hat{d}(a,b)<d+2g_{err}$ with high probability.
\end{lemma}

This means that the lack of error-independence will not cause LSH to miss internal nodes within specified evolutionary distance. On the other hand, it could happen that the bias from error-dependence generates additional hits in hash tables. The following lemma shows that the number of additional hits is at most $\log n$.
\begin{lemma}
Let $1/2 - 1/2(\exp(-2cf\log \log n)$ with $c<\ln 2$ in the LSH algorithm. The number of sequences $b$ such that $d(a,b)>cf\log\log n$ and $d(a,A(b))<cf\log\log n$ is at most $\log n$.
\end{lemma}
\begin{proof}
We use a well-known equivalent formulation of Markov chain on trees as a percolation process. For each edge in $T$, we set if to {\it open} with probability $1-2p$ and {\it closed} otherwise. Each connected component of open edges shares the same state and states in different components are independent. Notice that conditioned on there being a closed edge between $a$ and $b$, a reconstruction error in $b$ is independent of the state at $a$. If the true evolutionary distance between $a$ and $b$ is at least $f\log\log n$, the probability of them being in the same component is at most $(\log n)^{-2f}$. Consequently, the expected normalied Hamming distance between $a$ and $b$ is at least$$
\frac{1}{2}-(\log n)^{\frac{-2f}{\ln 2}}+(\log n)^{-2f}
$$
Picking $c<\ln 2$ ensures that this is bounded away from $r_2$ for $n$ large enough. This, together with the fact that there are at most $\log n$ sequences within distance $f\log\log n$ of $a$, concludes the proof.
\end{proof}


\end{document}